\documentclass[a4paper,10pt]{article}
\usepackage[font=small, labelfont={sc}]{caption}
\usepackage[a4paper,top=3.2cm,bottom=3.cm,left=2.8cm,right=2.8cm]{geometry}
\usepackage{latexsym}
\usepackage[latin1]{inputenc}

\usepackage{amsfonts,amssymb,epsfig,amsmath,amsthm}
\usepackage{setspace}
\usepackage{graphics,subfigure}
\usepackage{makeidx}
\usepackage{multicol}
\usepackage{hyperref}
\usepackage{color}

\numberwithin{equation}{section}

\def \beq {\begin{eqnarray}}
\def \eeq {\end{eqnarray}}
\def \beqn {\begin{eqnarray*}}
\def \eeqn {\end{eqnarray*}}

\newtheorem{theorem}{Theorem}[section]
\newtheorem{itlemma}[theorem]{Lemma}
\newtheorem{itproposition}[theorem]{Proposition}
\newtheorem{itcorollary}[theorem]{Corollary}
\newtheorem{itremark}[theorem]{Remark}
\newtheorem{itdefinition}[theorem]{Definition}
\newtheorem{itexample}[theorem]{Example}
\newtheorem{itclaim}[theorem]{Claim}
\newtheorem{itfact}[theorem]{Fact}
\newtheorem{itassumption}[theorem]{Assumption}

\newenvironment{fact}{\begin{itfact}\rm}{\end{itfact}}
\newenvironment{claim}{\begin{itclaim}\rm}{\end{itclaim}}
\newenvironment{lemma}{\begin{itlemma}}{\end{itlemma}}
\newenvironment{remark}{\begin{itremark}\rm}{\end{itremark}}
\newenvironment{corollary}{\begin{itcorollary}}{\end{itcorollary}}
\newenvironment{proposition}{\begin{itproposition}}{\end{itproposition}}
\newenvironment{definition}{\begin{itdefinition}\rm}{\end{itdefinition}}
\newenvironment{example}{\begin{itexample}\rm}{\end{itexample}}
\newenvironment{assumption}{\begin{itassumption}}{\end{itassumption}}

\newcommand{\be}[1]{\begin{equation}\label{#1}}
\newcommand{\ee}{\end{equation}}
\newcommand{\bl}[1]{\begin{lemma}\label{#1}}
\newcommand{\br}[1]{\begin{remark}\label{#1}}
\newcommand{\brs}[1]{\begin{remarks}\label{#1}}
\newcommand{\bt}[1]{\begin{theorem}\label{#1}}
\newcommand{\bd}[1]{\begin{definition}\label{#1}}
\newcommand{\bp}[1]{\begin{proposition}\label{#1}}
\newcommand{\bc}[1]{\begin{corollary}\label{#1}}
\newcommand{\bfact}[1]{\begin{fact}\label{#1}.}
\newcommand{\bex}[1]{\begin{example}\label{#1}}
\newcommand{\ec}{\end{corollary}}
\newcommand{\efact}{\end{fact}}
\newcommand{\eex}{\end{example}}
\newcommand{\el}{\end{lemma}}
\newcommand{\er}{\end{remark}}
\newcommand{\ers}{\end{remarks}}
\newcommand{\et}{\end{theorem}}
\newcommand{\ed}{\end{definition}}
\newcommand{\ep}{\end{proposition}}
\newcommand{\epr}{\end{proof}}
\newcommand{\bpr}{\begin{proof}}
\newcommand{\bcl}[1]{\begin{claim}\label{#1}}
\newcommand{\ecl}{\end{claim}}
\newcommand{\bas}[1]{\begin{assumption}\label{#1}}
\newcommand{\eas}{\end{assumption}}

\newcommand{\ecs}{\end{corollary}}
\newcommand{\eers}{\end{exercise}}
\newcommand{\eexs}{\end{example}}
\newcommand{\eems}{\end{example}}
\newcommand{\els}{\end{lemma}}
\newcommand{\eles}{\end{lemmaex}}
\newcommand{\ets}{\end{theorem}}
\newcommand{\eds}{\end{definition}}
\newcommand{\eps}{\end{proposition}}

\newcommand{\bi}{\begin{itemize}}
\newcommand{\ei}{\end{itemize}}
\newcommand{\ben}{\begin{enumerate}}
\newcommand{\een}{\end{enumerate}}

\def\vbar{\mathchoice{\vrule height6.3ptdepth-.5ptwidth.8pt\kern-.8pt}
   {\vrule height6.3ptdepth-.5ptwidth.8pt\kern-.8pt}
   {\vrule height4.1ptdepth-.35ptwidth.6pt\kern-.6pt}
   {\vrule height3.1ptdepth-.25ptwidth.5pt\kern-.5pt}}
\def\fudge{\mathchoice{}{}{\mkern.5mu}{\mkern.8mu}}
\def\bbc#1#2{{\rm \mkern#2mu\vbar\mkern-#2mu#1}}
\def\bbb#1{{\rm I\mkern-3.5mu #1}}
\def\bba#1#2{{\rm #1\mkern-#2mu\fudge #1}}
\def\bb#1{{\count4=`#1 \advance\count4by-64 \ifcase\count4\or\bba A{11.5}\or
   \bbb B\or\bbc C{5}\or\bbb D\or\bbb E\or\bbb F \or\bbc G{5}\or\bbb H\or
   \bbb I\or\bbc J{3}\or\bbb K\or\bbb L \or\bbb M\or\bbb N\or\bbc O{5} \or
   \bbb P\or\bbc Q{5}\or\bbb R\or\bbc S{4.2}\or\bba T{10.5}\or\bbc U{5}\or
   \bba V{12}\or\bba W{16.5}\or\bba X{11}\or\bba Y{11.7}\or\bba Z{7.5}\fi}}

\makeatletter
\newtheorem*{rep@theorem}{\rep@title} \newcommand{\newreptheorem}[2]{%
\newenvironment{rep#1}[1]{%
\def\rep@title{\bf #2 \ref{##1} }%
\begin{rep@theorem} }%
{\end{rep@theorem} } }
\makeatother
\newreptheorem{theorem}{Theorem}
\newreptheorem{lemma}{Lemma}

\def \R {{\mathbb R}}

\def \N {{\mathbb N}}
\def \PR {{\mathbb P}}
\def \E {{\mathbb E}}

\def \s {y}

\newcommand{\ba}[1]{\addtocounter{for}{1} \begin{eqnarray}\label{#1}}
\newcommand{\ea}{\end{eqnarray}}

\def\sqr#1#2{{\vcenter{\vbox{\hrule height .#2pt
                             \hbox{\vrule width .#2pt height#1pt \kern#1pt
                                   \vrule width .#2pt}
                             \hrule height .#2pt}}}}

\def\pmb#1{\setbox0=\hbox{#1}%
   \kern-.025em\copy0\kern-\wd0
   \kern.05em\copy0\kern-\wd0
   \kern-.025em\raise.0433em\box0 }
\def\sqr#1#2{{\vcenter{\vbox{\hrule height.#2pt
     \hbox{\vrule width.#2pt height#1pt \kern#1pt
   \vrule width.#2pt}\hrule height.#2pt}}}}

\def\s{\sigma}

\def\g{\gamma}

\def\a{\alpha}

\def\cal{\mathcal}

\newenvironment{myenumerate}{%
\begin{list}{\labelenumi}
	{%
	\setlength{\itemsep}{0.4em}%
	\setlength{\topsep}{0.5em}%
	\setlength\leftmargin{2.6em}%
	\setlength\labelwidth{2.15em}%
	\setlength{\labelsep}{0.45em}%
	\usecounter{enumi}%
	}%
	}%
{\end{list}
}

\renewenvironment{enumerate}{
\renewcommand{\theenumi}{\arabic{enumi}}%
\renewcommand{\labelenumi}{{\rm(\theenumi)}}%
\begin{myenumerate}}%
{\end{myenumerate}}

{\end{myenumerate}}

{\end{myenumerate}}

\newenvironment{myitemize}{%
\begin{list}{$\bullet$}%
 	{%
	\setlength{\itemsep}{0.4em}%
	\setlength{\topsep}{0.5em}%
	\setlength\leftmargin{2.6em}%
	\setlength\labelwidth{2.15em}%
	\setlength{\labelsep}{0.45em}%
	}%
	}%
{\end{list}}

\renewenvironment{itemize}{
\begin{myitemize}}%
{\end{myitemize}}

\def\dd{\mathrm{d}}

\newcommand{\Keywords}[1]{\par\noindent 
{\small{\em Keywords\/}: #1}}

\newcommand{\classification}[1]{\par\noindent 
{\small{\em 2010 Mathematics Subject Classification.\/}: #1}}

\title{A multivariate model for financial indices and an algorithm for detection of jumps in the volatility}

\author{ \textsc{Mario Bonino} \thanks{%
Dipartimento di Matematica Pura ed Applicata, Universit\'a  degli Studi di Padova, via Trieste 63, I-35121 Padova, Italy.
{\texttt{mario.bonino@outlook.com}}
. }\smallskip\\
\textsc{Matteo Camelia}\thanks{%
Dipartimento di Matematica Pura ed Applicata, Universit\'a  degli Studi di Padova, via Trieste 63, I-35121 Padova, Italy.
{\texttt{matteo.camelia@gmail.com}}. }\smallskip\\
\textsc{Paolo Pigato}\thanks{%
INRIA Nancy \& IECL Campus Scientifique, BP 70239,  54506 Vandoeuvre-L\'es-Nancy - Cedex, France \texttt{%
paolo.pigato@inria.fr} (corresponding author)}\smallskip\\}

\begin{document}
\maketitle
\begin{abstract}
\noindent
We consider a mean-reverting stochastic volatility model which satisfies some relevant stylized facts of financial markets. We introduce an algorithm for the detection of peaks in the volatility profile, that we apply to the time series of Dow Jones Industrial Average and Financial Times Stock Exchange 100 in the period 1984-2013. Based on empirical results, we propose a bivariate version of the model, for which we find an explicit expression for the decay over time of cross-asset correlations between absolute returns. We compare our theoretical predictions with empirical estimates on the same financial time series, finding an excellent agreement. 
\smallskip
\classification{60G44, 91B25, 91G70}
\end{abstract}

\section{Introduction}  

In the last two decades a number of researchers has shown an increasing interest in the field of economics and finance and their links with statistical mechanics. Many interesting phenomena arise when looking at financial data with mathematical tools coming from statistical physics, this being motivated by the fact that a financial market is somehow analogous to a physical ``complex system", being the result of the interactions of a huge number of agents. What we look at is not the behavior of the single agent, but some macroscopic quantity that we consider important. This new viewpoint has led to the discovery of some striking empirical properties, detected in various types of financial markets, considered now as \emph{stylized facts} of these markets. Examples of such facts are scaling properties, auto-similarity, properties of the volatility profile, such as peaks and clustering, and long range dependence.

In this paper we deal with the two last phenomena, but our point of view comes from mathematical finance more than statistical mechanics. We do not look at the microscopic behavior, but directly at the macroscopic quantities mentioned above, even though we can suppose that the large-scale phenomena under study have their origin in some small-scale interactions. For this purpose we work in the framework of continuous-time stochastic volatility models. More precisely, the market models that we use in this article are mean reverting stochastic volatility models, with a volatility driven by a jump process. This means that our process for the detrended log-price evolves through $dX_t=\s_t dB_t$, where $B$ is a Brownian motion and the volatility $\s=(\s_t)_{t\in \R}$ is the square root of the stationary solution of a SDE of the following form:
\begin{equation}
\label{stomeanrev}
d \s_t^2 = -f(\s_t^2) dt + d L_t.
\end{equation}
The function $f$, what we call ``mean reversion", has the role of pushing the volatility back to its equilibrium value, whereas $L=(L_t)_t$ is a process which models the noise in the volatility, and it is often taken as a pure jump process (see \cite{fasen2006extremal, kluppelberg2006continuous, Barndorff-Nielsen:2001}). If $\s$ is an independent process with paths in $L^2_{loc}(\R)$ a.s., the process $X$ can be viewed as a random time-change of Brownian motion: $X_t= W_{I(t)}$, where $I(t)=\int_0^t \s^2_s ds$. An example of such process is the model presented in \cite{acdp}, that will be the main focus of this article. The introduction of this model is justified by its capability to account a number of the stylized facts mentioned above, namely: the crossover in the log-return distribution from a power-law
to a Gaussian behavior, the slow decay in the volatility autocorrelation, the diffusive scaling and the multiscaling of moments, while keeping a simple formulation and explicit dependence on the parameters. We consider here the two following issues:
an algorithm for the detection of shocks in the market (peaks in the volatility profile) and a study of a bivariate version of the model. These two aspects are linked by the fact that the cross-correlation between two indices is highly dependent on the correlation between the jumps of the volatility processes.

After a sinthetic presentation of the model introduced in \cite{acdp}, we propose our algorithm for the detection of jumps in the volatility. The problem of finding shocks in financial time series is a classical one. For example, GARCH models (Generalized Autoregressive Conditional Heteroskedasticity, \cite{Bollerslev:1986}) are widely used, but in practice ``volatility seems to behave more like a jump process, where it fluctuates around some value for an extended period of time, before undergoing an abrupt change" (\cite{ross}). To adress this issue, regime-switching GARCH models have been developed (\cite{gray,he_maheu}), but they can be hard to implement. Therefore, a more common approach is to use an approximate procedure, the so-called ICSS-GARCH algorithm, introduced in \cite{icss}. 
When working in continuous time, a classical tool for estimating volatility is to use approximations of quadratic variation of the price process via sums of squared returns (see \cite{BNS}, or the recent book \cite{JacodProtter} for theoretical results). 
Here we work in this framework, with the specific aim of estimating shocks. Therefore, our estimator uses squared returns, as the ICSS-GARCH algorithm. However, the ICSS-GARCH algorithm works well under the assumption that the returns are normally distributed. Our algorithm, on the contrary, is based on geometrical considerations on the quadratic variation of the price process.

We use our algorithm on the empirical time series of the Dow Jones Industrial Average (DJIA) and the Financial Times Stock Exchange (FTSE) 100, from 1984 to 2013. Some heuristic considerations on the output confirm its validity in the detection of jumps. We find that the majority of the peaks of the volatility estimated by the algorithm are shared by the two indices.
Motivated by the fact that most of the shocks are common to the two markets, we consider a bivariate version of the model, where the joint process of shocks is given by correlated Poisson  processes. This is a main ingredient in our modeling, since the long range dependence heavily relies, through the volatility process, on the shock times. Indeed, defining
\begin{align*}
dX_t&=\s^X_t dB^X_t,\quad &d (\s^X_t)^2 = -f((\s^X_t)^2) dt + d L^X_t, \\
dY_t&=\s^Y_t dB^Y_t,\quad &d (\s^Y_t)^2 = -f((\s^Y_t)^2) dt + d L^Y_t,
\end{align*}
it is easy to show under very weak hypothesis on the joint volatility $(\s^X_t,\s^Y_t)_{t\in\R}$ that
$$
\lim_{h\downarrow 0} corr (|X_h-X_0|,|Y_{t+h}-Y_t|)=corr(\s^X_0,\s^Y_t).
$$ 
If the volatilities are of the precise form considered in \cite{acdp}, explicit computations are possible and the evolution of $(\s^X,\s^Y)$ depends only on the jumps of $L^X$ and $L^Y$. The correlation of both increments and absolute increments of two assets at a certain time has been widely studied, especially because of its direct link with systemic risk and portfolio management (see for instance \cite{Embrechts:2002, Brownlees:2012}), but here we deal with something more peculiar. We consider the cross-correlation of absolute increments at different times, and compute how this correlation decays as the time difference increases. This issue has been addressed by 
Podobnik et al. in \cite{Podobnik:2007}, where they analyze the Dow Jones Industrial and the S\&P500 indices, and in \cite{Podobnik:2010,Podobnik:2011}, where long range cross-correlations between magnitudes are found in a number of studies including nanodevices, atmospheric
geophysics, seismology and finance. In our framework, we find this explicit formula for the decay of cross-asset correlations between absolute returns, depending on the time lag (see Corollary \ref{corcorr}):
\[
\lim_{h\downarrow 0} corr (|X_h-X_0|,|Y_{t+h}-Y_t|)=\frac{2}{\pi} \frac{Cov\left((S^X)^{D^X-1/2},(\lambda^Y t+S^Y)^{D^Y-1/2}\right)}
{\sqrt{Var(|N| S^{D^X-1/2})Var(|N| S^{D^Y-1/2})}}e^{-\lambda^Y t}
\]
The quantities involved are constant parameters of the volatilities $\s^X$ and $\s^Y$, except from $S^X$ and $S^Y$, which are correlated exponential variables coming from the jump process $L=(L^X,L^Y)$, and the two independent random variables $N$ (standard Gaussian) and $S$ (exponential). We compare this result to the two empirical time series of FTSE and DJIA, finding an excellent agreement between predictions of the model and empirical findings. In particular we find that from both model and empirical data the decay of autocorrelations and cross-correlations is almost coincident, and in particular it is slow over time, confirming that this is a long-memory processes.
On the other hand, empirically we find a non-significant cross-correlation between returns of FTSE and DJIA, even for very small time lags, and this is consistent with the model as well. 
\begin{figure}[!ht]
\centering
\caption{Decay of cross-correlations}
\includegraphics[width=0.46\textwidth]{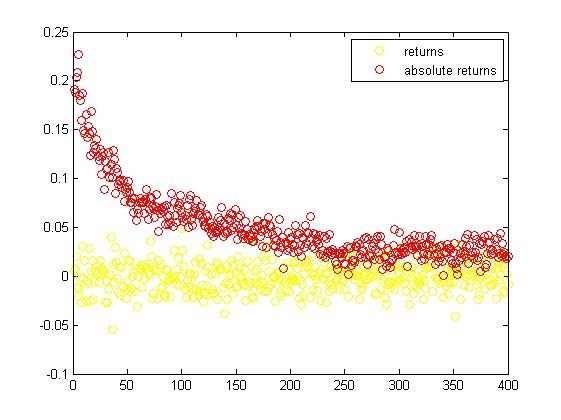}
\label{crosscorr}
\end{figure}
This is not surprising, since for both indices there are no long-range autocorrelations of returns, and this is easily seen to be consistent with our model. In contrast, as already said, the decay of cross-correlation of absolute returns is very slow (see Fig. \ref{crosscorr}). 

The fact that our model behaves as real financial time series in all of these aspects is a remakable validation. We mention that the satistical analysis performed by Podobnik et al. in \cite{Podobnik:2007,Podobnik:2011} lead to results analogous to ours, concerning also the similarity in the decay of autocorrelations and cross-asset correlations.

In section \ref{sectdef} we present the model introduced in \cite{acdp} and explain the algorithm for the detection of jumps in the volatility. In section \ref{sectcorr} we study the bivariate model.
In section \ref{sectproofs} we prove the results motivating the algorithm and the results on the bivariate model.

\section{Definition of the model and detection of jumps}
\label{sectdef}

In this chapter we describe the model and state properties and results related to stylized facts. 
Then, based on properties of the quadratic variation of our price process, we propose an algorithm for the detection of shocks in the market.

\subsection{Definition of the model}
\label{def}

Given three real numbers $D\in(0,1/2]$, $\lambda\in(0,\infty)$, $\bar{\s}\in(0,\infty)$, the model is defined upon two sources of randomness:
\begin{itemize}
\item a Brownian motion $W=\left(W_t\right)_{t\geq 0}$;
\item a Poisson point process $\mathcal{T}=(\tau_n)_{n\in\mathbb{Z}}$ of rate  $\lambda$ on $\mathbb{R}$.
\end{itemize}
We suppose $W$ and $\mathcal{T}$ independent. By convention we label the points of $\mathcal{T}$ so that $\tau_0<0<\tau_1$. For $t\geq 0$, we define
$$
i(t) := \sup\{n\geq 0 : \tau_n \leq t\}= \# \{\mathcal{T} \cap (0, t]\}.
$$
$i(t)$ is the number of positive times in the Poisson process before $t$, so that $\tau_{i(t)}$ is the location of the last point in $\mathcal{T}$ before $t$. We introduce the process $I=(I_t)_{t\geq 0}$ defining
\begin{equation}
\label{timechange}
I_t=\bar{\s}^2\left[(t-\tau_{i(t)})^{2D}+\sum^{i(t)}_{k=1}(\tau_k-\tau_{k-1})^{2D}-(-\tau_0)^{2D}\right]
\end{equation}
where we agree that the sum in the right hand side is $0$ if $i(t)=0$. Now we define the process which is the model for the detrended log price as
\begin{equation}
X_t=W_{I(t)}.
\label{def1}
\end{equation}
Observe that $I$ is a strictly increasing process with absolutely continuous paths, and it is independent of the Brownian motion $W$. Thus this model may be viewed as an independent random time change of a Brownian motion.

We shortly give a motivation for this definition. Remark that for $D=1/2$ the model reduces to Black \& Scholes with volatility $\bar{\s}$. For $D<1/2$, the introduction of a time inhomogeneity $t\rightarrow t^{2D}$ at times $\tau_n$ is meant to represent the \emph{trading time} of a financial time series, where at random times there are shocks in the market, modeled by our Poisson point process. The reaction of the market is an acceleration of the dynamics immediately after the shock, and a gradual slowing down at later times, until a new shock accelerates the dynamics again. This behavior is due to the function $t\rightarrow t^{2D}$, $D\in(0,1/2]$, which is steep for $t$ close to $0$ and bends down for increasing $t$.

\begin{figure}[htbp]
\caption{Time inhomogeneity}
\begin{center}
\includegraphics[width=7.cm]{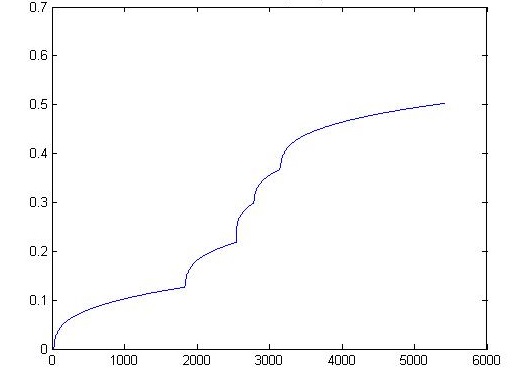}
\end{center}
\end{figure}

The definition of the model as a time changed Brownian Motion implies that we can equivalently express it as a a stochastic volatility model, where the volatility is
\[
\s_t=\sqrt{I'(t)}=\sqrt{2D}\bar{\s} (t-\tau_{i(t)})^{D-1/2},
\]
and the evolution of $X$ is given by $dX_t=\s_t dB_t$. To write the volatility as solution of a stochastic differential equation of the form \eqref{stomeanrev}, we can define it as the the stationary solution of
\[
d(\s_t^2)=-\alpha (\s_t^2)^\gamma dt +\infty \,d i(t), 
\]
where the constants are
$$
\gamma=2+\frac{2D}{1-2D}>2,\quad \alpha=\frac{1-2D}{(2D)^{1/(1-2D)}}\frac{1}{\bar{\s}^{2/(1-2D)}}.
$$
This process is well defined, since after the infinite jumps the super-linear drift term instantaneously produces a finite pathwise solution. We refer to \cite{articolomultiscaling} for the details of the correspondence between time change and stochastic volatility in this framework, for a wider class of stochastic volatility models.

\br{}
In the most general version of this model the parameter $\bar{\s}$ is not constant. A sequence of random variables $(\bar{\s}_n)_{n\in\N}$ is simulated, and each of them is associated to the corresponding jump. The results presented here are still valid in this case, with a slightly more complicated formulation. We assume $\bar{\s}$ constant in this work, since calibration on data coming from financial time series leads in any case to this type of choice. \er

\subsection{Main properties}
We briefly recall some properties of the process $X$.
For proofs, more detailed statements and some additional considerations we refer to \cite{acdp}.
\begin{proposition}[Basic Properties] \label{basicpr}
Let $X$ be the process defined in \eqref{def1}. The following assertions hold:
\begin{enumerate}
  \item $X$ has stationary increments.
\item $X$ is a zero-mean, continuous, square-integrable martingale, with quadratic variation $\langle X\rangle_t=I_t$.
\item The distributions of the increments of $X$ is ergodic.
\item 
$\E(|X_t|^q) < \infty$ for some (and hence any) $t > 0,\,q\in [0,\infty)$.
\end{enumerate}
\label{procprop}
\end{proposition}
We are now ready to state some results, important because they establish a link between our model and the stylized fact mentioned in the introduction. The process $X$ defined in \eqref{def1} is consistent with important facts empirically detected in many  real time series, namely: diffusive scaling of returns, multiscaling of moments, slow decay of volatility autocorrelation. 

The first result shows that the increments $(X_{t+h} - X_t)$ have an approximate diffusive scaling both when $h\downarrow 0$, with a heavy-tailed limit distribution,
and when $h \uparrow \infty$, with a normal limit distribution. This is a precise mathematical
formulation of a crossover phenomenon in the log-return distribution,
from approximately heavy-tailed (for small time) to approximately Gaussian (for large time).

\begin{theorem}[Diffusive scaling] \label{th:scaling}
The following convergences in distribution hold
for any choice of the parameters $D, \lambda, \bar{\s}$.
\begin{itemize}
\item Small-time diffusive scaling:
\begin{align} \label{eq:convdist}
	\frac{(X_{t+h} - X_t)}{\sqrt h} \ \xrightarrow[\ h \downarrow 0\ ]{d} \
	f(x) \, \dd x \ := \ \text{law of } \ \bar{\s}\big(\sqrt{2D} \,
	\lambda^{\frac{1}{2} - D} \big) \, S^{D-\frac{1}{2}} \, W_1 \,,
\end{align}
where
$S \sim Exp(1)$ and $W_1 \sim \mathcal{N}(0,1)$
are independent random variables. The function $f$ is 
$$
f(x)=\int_0^\infty dt\, \lambda e^{-\lambda t} \frac{t^{1/2-D}}{\bar{\s}\sqrt{4D\pi}}exp\left(-\frac{t^{1-2D}x^2}{4D\bar{\s}^2}\right).
$$

\item Large-time diffusive scaling:
\begin{align} \label{eq:convdist2}
	\frac{(X_{t+h} - X_t)}{\sqrt h} \ \xrightarrow[\ h \uparrow \infty\ ]{d} \
	\frac{e^{-x^2/(2c^2)}}{\sqrt{2\pi} c} \, \dd x \, = \, \mathcal{N}(0,c^2) \,,
	\qquad
	c^2 =  \bar{\s}^2 \,\lambda^{1-2D}\, \Gamma(2D+1) \,,
\end{align}
where $\Gamma(\alpha) := \int_0^\infty x^{\alpha-1} e^{-x} \dd x$
denotes Euler's Gamma function.
\end{itemize}
\end{theorem}

The density $f$, when $D < \frac{1}{2}$, has \emph{power-law tails}:
\[
\E_f(|x|^q)=\infty \Leftrightarrow q\geq q^*:=(1/2-D)^{-1}. 
\]
The function $f$, which describes the asymptotic law, for $h\downarrow 0$, of $\frac{X_{t+h}-X_t}{\sqrt{h}}$, has a different tail behavior from the density of $X_{t+h}-X_t$, for fixed $h$ (cf. Proposition \ref{procprop} point 4).
This feature of $f$ is linked to another property of our model: the multiscaling of moments. Let us define the $q-th$ moment of the log returns, at time scale $h$:
\[
m_q(h):=\E\left(|X_{t+h}-X_t|^q\right)=\E\left(|X_h|^q\right)
\] 
the last equality holding for the stationarity of the increments.
Because of the diffusive scaling properties (Theorem \ref{th:scaling}), we would expect $m_q(h)$ to approximate in some sense $h^\frac{q}{2} \int x^q f(x) dx = C_q\, h^\frac{q}{2}$, for $h\downarrow 0$. This is actually true for $q< q^*$, that is, for $q$ such that the $q-th$ moment of the limit distribution is finite. For $q\geq q^*$, the $q-th$ moment of the limit distribution is not finite, and it turns out that a faster scaling holds, namely $m_q(h)\approx h^{Dq+1}$. This transition in the scaling of $m_q(h)$ is known as multiscaling of moments, a property empirically detected in many time series, in particular in financial series. The following theorem states that for  this model the multiscaling exponent is a piecewise linear function of $q$. 
In \cite{articolomultiscaling} the problem of multiscaling in more general stochastic volatility models is considered, finding that an analogous behavior is common to a much wider class.

\begin{theorem}[Multi-scaling of moments]
\label{th:multi}
For $q>0$ the $q-th$ moment of log returns $m_q(h)$ has the following asymptotic behavior as $h\downarrow 0$:

\[
m_q(h)\sim\left\{
\begin{aligned}
&C_q h^\frac{q}{2},&              \quad               &if\, q<q^*\\
&C_q h^\frac{q}{2}\log\left(\frac{1}{h}\right),&      &if\, q=q^*\\
&C_q h^{Dq+1},&                                       &if\, q>q^*
\end{aligned}\right.
\]
for some constants $C_q \in (0,\infty)$ (whose explicit expression can be found in \cite{acdp}).
As a consequence, the scaling exponent $A(q)$ is
\[
A(q)=\lim_{h\downarrow 0} \frac{\log m_q(h)}{\log h}=
\left\{
\begin{aligned}
&\frac{q}{2}&\, &if\, q\leq q^*\\
&Dq+1&    &if\, q\geq q^*
\end{aligned}\right.
\label{computeA}
\]
\end{theorem}
We now state a result concerning the volatility autocorrelation of the process $X$, that is the correlations of absolute values of returns at a given time distance. Recall that the correlation coefficient of two random variables $X$ and $Y$ is 
$$
\rho(X,Y)=\frac{Cov(X,Y)}{\sqrt{Var(X)Var(Y)}}.
$$
For the process $X$, introduce $\xi=(\xi_t)_{t\geq 0}$, the process of absolute values of increments, for $h$ fixed: $\xi_t=|X_{t+h}-X_t|$. Then the volatility autocorrelation of $X$ is
\[
\rho(t-s)=\lim_{h\downarrow 0}\rho(\xi_s,\xi_t)=\frac{Cov(\xi_s,\xi_t)}{\sqrt{Var(\xi_s)Var(\xi_t)}}
\]
Indeed, being the process stationary, the quantity we have defined above depends only on the time difference $t-s$.

\begin{theorem}[Volatility autocorrelaton]
\label{theoautocorr}
For $t \geq 0$,
\[
\rho(t)=
\frac{2}{\pi} \frac{Cov\left(S^{D-1/2},(\lambda t+S)^{D-1/2}\right)}{Var(|N| S^{D-1/2})}e^{-\lambda t}
\]
where $S$ is an exponential variable with parameter $1$, $N$ is a standard normal variable and they are mutually independent.
\label{volatility}
\end{theorem}
This theorem shows that the decay of volatility autocorrelation is between polynomial and exponential for $t=O(1/\lambda)$, exponential for $t>>1/\lambda$.

\subsection{Jumps and quadratic variation}
\label{heualgo}

We consider now the \emph{quadratic variation} of the price process, also referred to as the \emph{integrated volatility}. The aim of the following considerations is to introduce an algorithm for the detection of relevant jumps in the volatility, which has been proposed for the first time in \cite{tesi_mario}, 
and then analyzed in \cite{tesi_matteo,tesi_paolo}. 
 We know that the quadratic variation of $X$ is given by $I$ (Proposition \ref{basicpr}):
 \[
\langle X \rangle_{t} = I_{t}.
 \]
We recall that the quadratic variation is the limit in probability of the sum of squared increments on shrinking partitions. Therefore, a natural estimator of $I$ is the sum of squared increments of a dense sampling of $X$  (see \cite{JacodProtter} for many results on consistence and speed of convergence of this kind of estimators). 

On the other hand, the process $I$ is piecewise-concave; in fact, we recall that such process is defined by
\[
I_t=\bar{\s}^2\left[(t-\tau_{i(t)})^{2D}+\sum^{i(t)}_{k=1}(\tau_k-\tau_{k-1})^{2D}-(-\tau_0)^{2D}\right]
\]
It is clear that between two consecutive shock times the process is concave. We fix time $T$ and consider the backward difference quotient defined by
\[
Q_{T}(t) := \frac{I_{T} - I_{T-t}}{t}.
\]
This quantity is, conditional on $\mathcal{T}$, increasing up to the last shock time before $T$, and therefore it has a local maximum in $ t = T - \tau_{i(T)}$. Moreover, the derivative of $I_{t}$ is very large after a shock, but it quickly decays over time. Because of that, we expect that $Q_T(s) < Q_T(T-\tau_{i(T)})$ if $ s \in (T-\tau_{i(T)}, T-\tau_{i(T)}-L )$, for some $L >0$. We propose here an algorithm based on the following idea: if we choose $M > 0$  such that $\tau_{i(T) -1}<T-M < \tau_{i(T)} < T$ and $T-M$ is ``closer to $\tau_{i(T)}$ than to $\tau_{i(T) -1}$'', then the global maximum of $Q_T(t)$ in the interval $(T-M, T)$ should be attained at $t = T-\tau_{i(T)}$.  The geometric properties of $Q_T$ motivating this fact are stated and proved in \ref{lemma:geom}. Since a natural estimator for the quadratic variation of $X$ is an average of squared increments of $X_{t}$, we introduce the following estimator:
\[
V_{T}(k) := \frac{1}{k} \sum_{i=1}^{k} \left( X_{T-i+1}-X_{T-i} \right)^{2}
\]
Conditioning on $\cal{T}$, we have
\begin{eqnarray*}
\E \left \{ V_{T}(k) \vert \mathcal{T} \right \} &=& \frac{1}{k} \ \sum_{i=1}^{k} \E \left \{ \left( X_{T-i+1}-X_{T-i} \right)^{2} \vert \mathcal{T}\right \} \\
&=& \frac{1}{k} \ \sum_{i=1}^{k} \E \left \{ \left( W_{I_{T-i+1}}-W_{I_{T-i}} \right)^{2} \vert \mathcal{T} \right \} \\
&=& \frac{1}{k} \ \sum_{i=1}^{k} \left( I_{T}-I_{T-k} \right) \\
&=& Q_{T}(k)
\end{eqnarray*}
Since we cannot observe $Q_T$ directly on historical data, we use $V_{T}$ as an approximation and implement an algorithm for finding the realized jump times.

\subsection{Detection of jumps in the volatility}
\label{algo}
We describe now implementation and application of the algorithm.
We start introducing some notation. The financial index time series will be denoted by $(s_{i})_{0 \leq i \leq N}$, whereas the detrended logarithmic time series will be indicated by $(x_{i})_{250 \leq i \leq N}$, where 
\[
x_{i} := \log(s_{i}) - \bar{d}(i)
\]
and $\bar{d}(i) := \frac{1}{250} \sum_{k=i-250}^{i-1} \log(s_{i})$; we observe that it is not possible to define $x_{i}$ for $i < 250$. We  define $(y(i))_{0 \leq i \leq N}$ to be the corresponding series of trading dates. We also introduce the empirical estimate of $V_{N}$ as
\[
\widehat{V_{N}}(k)  := \frac{1}{k} \sum_{i=1}^{k} \left( x_{N-i+1}-x_{N-i} \right)^{2}
\]
Now, suppose that we want to know when the last shock in the time series occurred. The idea is to choose an appropriate integer $M$ such that $0 < M \leq N$ and see where the sequence $(\widehat{V_{N}}(k))_{N-M \leq k \leq N}$ attains its maximum. This leads us to the following definition.

\begin{definition}
    Let $(s_{i}), (x_{i}), (y_{i}), N, M$ be as above; given an integer $\tilde{N}$ such that $M \leq \tilde{N} \leq N$, we define 
    \[
        \widehat{k}(\tilde{N}, M) := \text{argmax}_{\tilde{N}-M \leq k \leq \tilde{N}} \quad \frac{1}{k} \sum_{i=1}^{k} \left( x_{\tilde{N}-i+1}-x_{\tilde{N}-i} \right)^{2}.
    \]
This quantity is an estimate of the distance of the last shock time before $y_{\tilde{N}}$ from $\tilde{N}$. We define also
    \[
        \widehat{i}(\tilde{N}, M) := \tilde{N} - \widehat{k}(\tilde{N}, M) + 1,
    \]
our estimate of the index of the last shock time estimate, and consequently our estimate of the last shock time before $y_{\tilde{N}}$ is
    \[
        \widehat{\tau}(\tilde{N}, M) := y \left( \widehat{i}(\tilde{N}, M) \right).
    \]
\label{def_estimate}
\end{definition}
It is worth comparing our algorithm with the so called ICSS-GARCH algorithm. Following \cite{ross}, we can describe the ICSS-GARCH algorithm as follows. Given a series of financial returns $r_{1}, \dots , r_{n}$, with mean $0$ we define the cumulative sum of squares $C_{k} = \sum_{i=1}^{k} r^{2}_{i}$ and let
\[
D_{k} = \frac{C_{k}}{C_{n}} - \frac{k}{n}, \qquad 1 \leq k \leq n , \qquad D_{0} = D_{n} = 0
\]
The idea
is that if the sequence $r_{1}, \dots , r_{n}$ has constant variance, then the sequence $D_{1}, \dots , D_{n}$ should oscillate around $0$. However, if there is a shock in the variance, the sequence should exhibit extreme behavior around that point. 

We remark that both algorithms use sums of squared returns to detect volatility shocks. However, the ICSS-GARCH algorithm works well under the assumption that the returns are normally distributed,
but not with heavy-tailed distributions, as proved in \cite{ross}. 
On the contrary, our algorithm is based on geometrical considerations: it locates shocks from the fact that our returns are given by a piecewise-concave time change  of a Brownian motion. The assumption of a piecewise-concave time change is reasonable in the context of stochastic volatility models. Indeed, to reproduce jumps in the volatility, one can introduce a process that makes the volatility strongly increase when a shock occurs, and then slowly decay over time.

For the empirical discussion outlined here, we use the DJIA Index and FTSE Index, from April 2nd, 1984 to July 6th, 2013, so that $N = 7368$. A similar data analysis has been done on the Standard \& Poor's 500 Index, from January 3rd, 1950 to July 23th, 2013, finding analogous results and confirming the validity of the method on aggregate indices.
All the calculations and pictures presented here have been obtained using the software MatLab \cite{MATLAB:2009}. An example of the empirical procedure to estimate the last shock time is given in Figure \ref{1tempofin}.

\begin{figure}[ht!]
    \caption{Plot of the quantity $\widehat{V_{\tilde{N}}}(k)$ for $k = 1, \dots, 2000$ ($M = 2000)$. $\tilde{N}$ has been choosen so that $y_{\tilde{N}}$ is the 10th of May 2011. The peak corresponds to the the 15th of September 2008, the day of the Lehman Brothers bank bankruptcy.}
    \begin{center}
      \includegraphics[scale=.4]{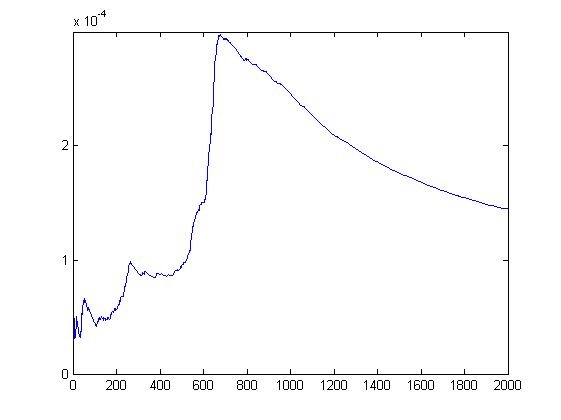}
    \end{center}
    \label{1tempofin}
\end{figure}
To confirm that the  estimate of the shock time is good, we may repeat the procedure approaching the shock time, dropping the last observation, or dropping a number of the last observations. If the estimate of the last shock time is confirmed, then we have a clear indication of the presence of a shock there (see Figure \ref{sovr}--\subref{dj_confirmed}). We remark that, when more than one shock is present on the considered time interval, the most recent is found as the maximum peak of $\hat{V}$ if we take $y_{\tilde{N}}$ close enough to it (again, this is motivated by lemma \ref{lemma:geom}). When we get further, the chosen peak is not necessarily the most recent, as we can see in Figure \ref{sovr}--\subref{2picchi}.
\begin{figure}[!ht]
\centering
\caption{Plot of the quantities $\widehat{V_{\tilde{N}}}(k)$ for $k = 1, \dots, 2000$ ($M = 2000)$, for the DJIA. In each figure we shift $\tilde{N}$ 4 times of 20 working days. 
In \subref{dj_confirmed} $\tilde{N}$ has been chosen so that $y_{\tilde{N}}$ is the 10/05/11(red), the 11/04/11(yellow), the 14/03/11(green) and the 11/02/11(blue). The four maxima are all located the 15/09/08, the day of the Lehman Brothers bank bankruptcy, confirming the presence of a shock there.
In \subref{2picchi} $\tilde{N}$ has been chosen so that $y_{\tilde{N}}$ is the 27/02/12(red), the 27/01/12(yellow), the 28/12/11(green) and the 29/11/11(blue). We can see that when $y_{\tilde{N}}$ is close to the 05/08/11 (European sovereign debt crisis), this date corresponds to the maximum of $\hat{V}$, whereas when we move further the maximum is again on the 15/09/08.}
\label{sovr}
\subfigure[]{
\includegraphics[width=0.46\textwidth]{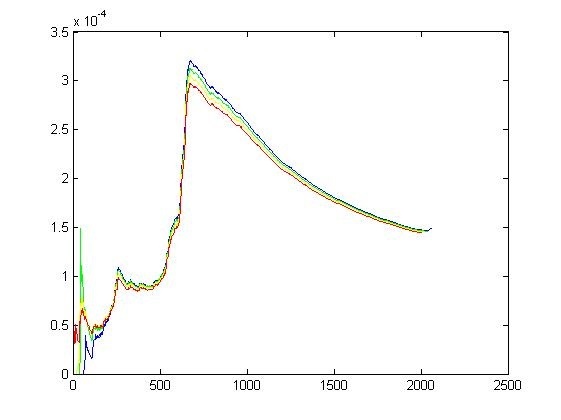}
\label{dj_confirmed}}
\hspace{1mm}
\subfigure[]{
\includegraphics[width=0.46\textwidth]{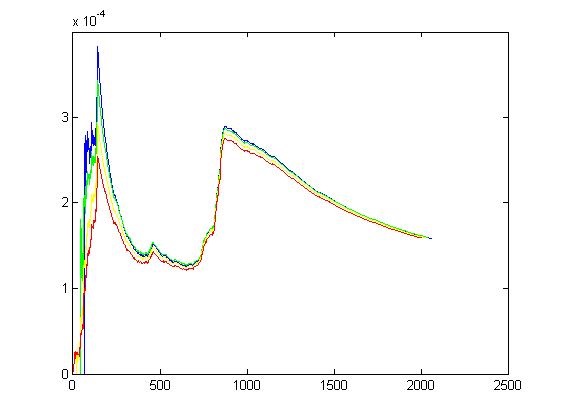}   
\label{2picchi}}
\end{figure}

In order to get a clearer result, we slightly tweak the procedure. When calculating $\widehat{k}(\tilde{N}, M)$ we ignore the last $20$ elements of the sum. In order words, instead of calculating $\widehat{k}(\tilde{N}, M)$ as the argmax for $\tilde{N}-M \leq k \leq \tilde{N}$, we drop the last $20$ elements of the series. This way we can locate only shocks that are at least $20$ days old, but we remove the large fluctuations that can be present immediately after a shock, which could cause some instability in the algorithm.

To locate all the past shocks in a given time series, we calculate the quantity $\widehat{i} (\tilde{N}, M)$ for $\tilde{N} = N, \dots, M$, and introduce the following sequence.

\begin{definition}
Given the quantities in definition \ref{def_estimate}, we introduce the shock time sequence as
\[
\widehat{h}((x_{i})_{250 \leq i \leq N}, M) := \left(  \widehat{i}(\tilde{N}, M)  \right)_{ M \leq \tilde{N} \leq N }
\]
\end{definition}

Finally, to get a clear picture of where the relevant shocks are in the time series, we can plot the number of occurrences of each element of the sequence $\widehat{h}((x_{i})_{250 \leq i \leq N}, M)$. We may choose to consider a date to be a shock-date if its number of occurrences exceeds a certain threshold. Table \ref{dates} contains our estimated shock-dates.
In Figure \ref{tau_sep} one can see the graphical evidence that maxima are concentrated on a small set of days for both FTSE and DJIA, supporting the validity of the method. The choice of the threshold is not completely determined, and we have based it on two criteria. Firstly, the number of estimated shocks should be consistent with the number of expected jumps of the Poisson process (whose rate is calibrated in section \ref{sectcorr}). Secondly, we see that in both series every date has a number of occurrences which is small or very large.
More explicitly, for the DJIA there are just 3 dates found approximately 50 times, whereas all the others are found more than 80 times or less than 25. Analogously, for the FTSE there are just 2 dates found approximately 50 times, whereas all the others are found more than 80 times or less than 20. Therefore it is reasonable to consider true shocks the ones with more than 80 occurrences, whereas it is not that clear how to consider the dates with approximately 50 occurrences. In any case, these choices are consistent with the number of expected jumps of the Poisson process.
Another issue in the choice of shock-dates is the fact that sometimes two or more very close dates are found a considerable number of times. In this case we consider them as related to the same shock. These dates are marked with the word "sparse" in table \ref{dates}, where we have reported our estimated dates.

\begin{figure}[!ht]
\centering
\caption{Shock times; x-axis: increasing time index; y-axis: y(i)=number of times the maximum of $\hat{V}$ is realized at i}
\subfigure[FTSE shock times]{
\includegraphics[width=0.46\textwidth]{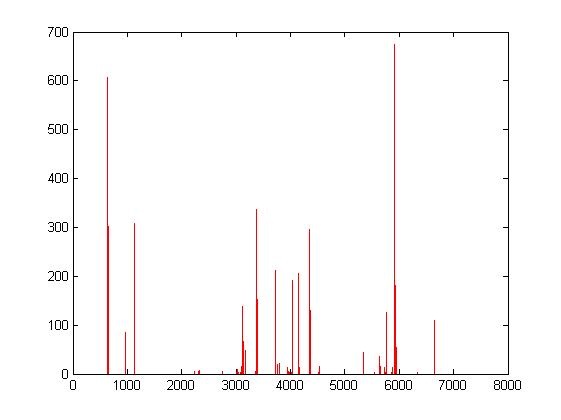}}
\hspace{1mm}
\subfigure[DJIA shock times]{
\includegraphics[width=0.46\textwidth]{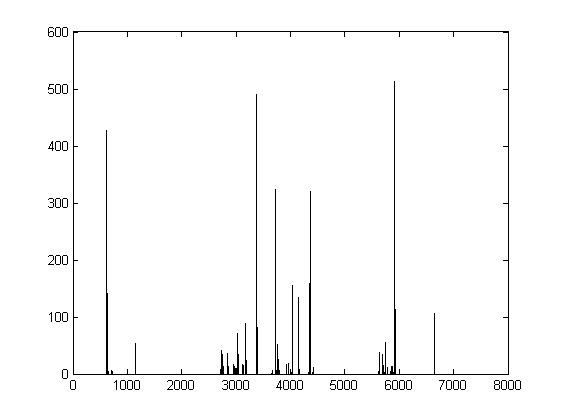}}
\label{tau_sep}
\end{figure}

\br{rs} It is natural at this point to wonder if there is a relation between the shocks in the two indices, and a straightforward experiment is to try to superimpose the two graphics (see Figure \ref{tau_sovr}). What we get is a clear indication that the shock times of the two series are almost coincident, only the magnitude (or evidence) being different and having very few shocks which are present only in one of the two indices. This is an important hint in the choice of the volatility for the bivariate process, a problem tackled in the following section.
\er

\begin{figure}[htbp]
\caption{Common jumps: overlap of Figures \ref{tau_sep} (a) and (b)}\begin{center}
\includegraphics[width=11.cm]{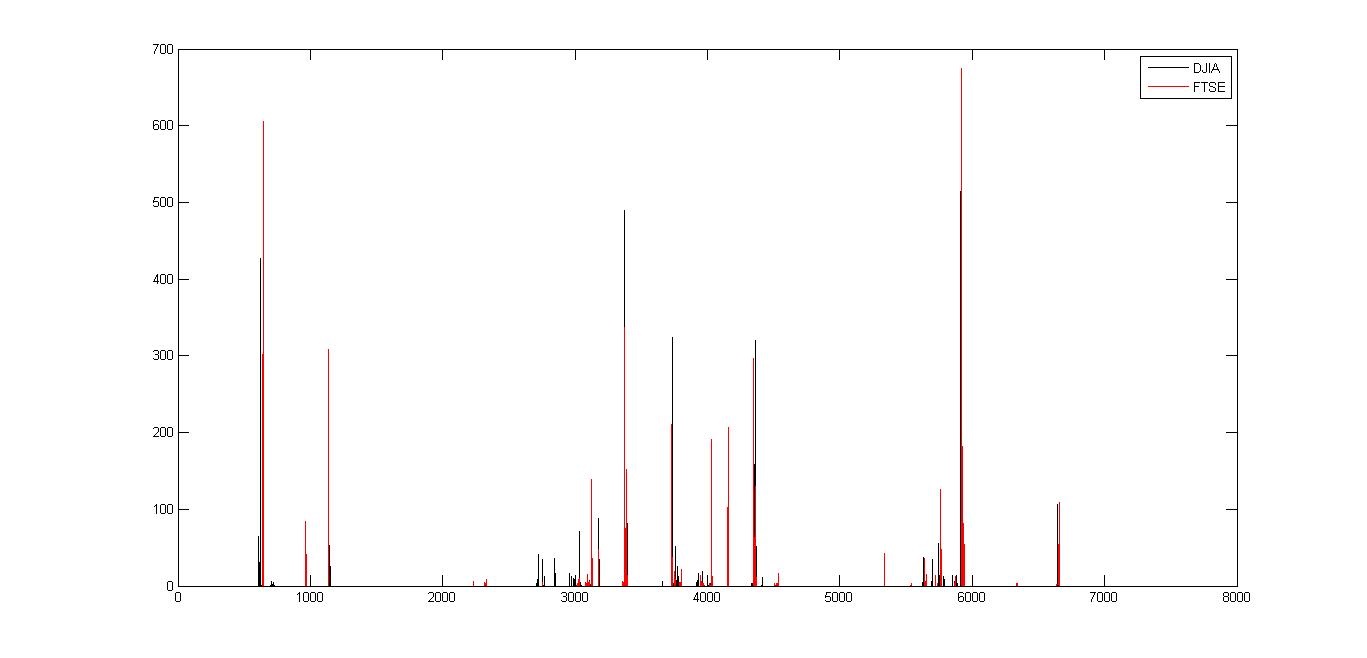}
\end{center}
\label{tau_sovr}
\end{figure}

\begin{table}
\caption{Estimated dates of shock times}
\begin{center}
\begin{tabular}{ll} \hline
{\em FTSE} & {\em DJIA} \\ \hline
14/10/87            					&   15/09/87  \\
24/01/89            					&              \\
26/09/89            					&   11/10/89  (questionable, 53)\\
                    					&   09/01/96  \\
                    					&   01/07/96  (questionable, 54)\\
                    					&   13/03/97  (sparse)\\
08/08/97            					&              \\
22/10/97 (questionable, 48)      &   16/10/97  \\
04/08/98            					&   31/07/98  \\
30/12/99            					&   04/01/00  \\
09/03/01            					&   09/03/01  \\
06/09/01            					&   06/09/01  \\
12/06/02            					&   05/07/02  \\
12/05/07 (questionable, 43)   &              \\
24/07/07            					&   24/07/07 (questionable, 57) \\
15/01/08            					&   04/01/08 (sparse) \\
03/09/08            					&   15/09/08  \\ 
05/08/11            					&   05/08/11  \\ \hline
\end{tabular}
\end{center}
\label{dates}
\end{table}

\section{The bivariate model}
\label{sectcorr}

\subsection{Definition of the bivariate model}
We investigate here the decay of the correlation between absolute returns of a bivariate version $(X,Y)=(X_t,Y_t)_{t\geq 0}$ of the model defined in Section \ref{sectdef}. We need the following quantities:  
\begin{itemize}
\item two Brownian motions $W^X=\left(W^X_t\right)_{t\geq 0}$ and $W^Y=\left(W^Y_t\right)_{t\geq 0}$;
\item two Poisson point processes on $\mathbb{R}$: $\mathcal{T}^X=(\tau^X_n)_{n\in\mathbb{Z}}$ and $\mathcal{T}^Y=(\tau^Y_n)_{n\in\mathbb{Z}}$, of rates respectively $\lambda^X$ and $\lambda^Y$;
\item positive constants $D^X$, $D^Y$, $\bar{\s}^X$ and $\bar{\s}^Y$.
\end{itemize}
The tricky point is the definition of the Poisson processes, that we want dependent but different. We introduce $\mathcal{T}^i,\,i=1,2,3$ independent Poisson processes with intensities $\lambda_i,\,i=1,2,3$. Then we define $\mathcal{T}^X=\mathcal{T}^1\cup \mathcal{T}^2$, $\mathcal{T}^Y=\mathcal{T}^1\cup\mathcal{T}^3$.
These are again Poisson processes, with intensity $\lambda_1+\lambda_2$ and $\lambda_1+\lambda_3$, and they are mutually dependent if $\mathcal{T}^1$ is non-degenerate.

We want to have a correlation coefficient $\rho\in [-1,1]$ also between the Brownian motions, so we introduce two independent Brownian motions $W^X,\,\tilde{W}$, and define 
\begin{align*}
W^Y_t=  \rho W^X_t + \sqrt{1-\rho^2} \tilde{W}_t.
\end{align*}
The correlation between $W^Y$ and $W^X$ will play no role in this paper, but the parameter $\rho$ is important for the correlation of the increments of $X$ and $Y$ at the same time, which could be an interesting aspect to consider. 
 
We suppose that the two-dim Brownian $W=(W^X,W^Y)$ and the two-dim time change $\mathcal{T}=(\mathcal{T}^X,\mathcal{T}^Y)$ are independent. The requirements of section \ref{sectdef} on the marginal one-dim processes are satisfied and we can define $X$ and $Y$ as
$$
X_t=W^X_{I^X_t},\quad Y_t=W^Y_{I^Y_t}
$$
where the random time changes $I^X_t$ and $I^Y_t$ are defined as in (\ref{timechange}). This definition is motivated by the fact that in empirical data the occurrence of a shock in one of the two indices often coincides with a peak in the volatility of the other one, as we noticed in remark \ref{rs}. So it is reasonable to suppose that part of the shock process is "common".

\subsection{Covariance and correlation of absolute log-returns}

For a given time $h$, we set
$$
\xi_t=|X_{t+h}-X_t|,\quad \eta_t=|Y_{t+h}-Y_t|,
$$
the absolute values of the returns of $X$ and $Y$ at time $t$. The following result on the asymptotic behavior of the covariance between $\xi_s$ and $\eta_t$ as the time lag $h$ goes to $0$ is proved in Section \ref{chaptproofs}.

\begin{theorem}[Covariance of absolute log-returns]
\label{thm:covar}
Let the process $(X,Y)$ be defined as above. Then, for any $t> s > 0$, the following holds: 
\[
\lim_{h\downarrow 0} \frac{Cov(\xi_s,\eta_t)}{h}=\frac{4\,\bar{\s}^X\bar{\s}^Y\,\sqrt{D^X D^Y}}{\pi}
Cov\left((-\tau^X_0)^{D^X-1/2},({t-s}-\tau^Y_0)^{D^Y-1/2}\right)e^{-\lambda^Y {(t-s)}}
\]
\end{theorem}
\br{alalala}
\label{unionpoi}
Using the definition of $\mathcal{T}^X$ and $\mathcal{T}^Y$ and the properties of Poisson processes it is possible to rewrite this expression as 
\begin{align*}
\lim_{h\downarrow 0} \frac{Cov(\xi_s,\eta_{t})}{h}=&
\frac{4}{\pi} \bar{\s}^X\bar{\s}^Y\,\sqrt{D^X D^Y}\left(\lambda^X\right)^{1/2-D^X}\left(\lambda^Y\right)^{1/2-D^Y}\times\\
&Cov\left((S^X)^{D^X-1/2},(\lambda^Y (t-s) +S^Y)^{D^Y-1/2}\right)e^{-\lambda^Y (t-s)}.
\end{align*}
Here $S^X$ and $S^Y$ are correlated exponential variables of parameter 1, defined as follows.
We set
\begin{align*}
& (\lambda_1+\lambda_2)S^{1,X}=(\lambda_1+\lambda_3)S^{1,Y} \sim 
\exp\left(\lambda_1\right), \\
& S^2\sim \exp\left(\frac{\lambda_2}{\lambda_1+\lambda_2}\right),
\quad S^3\sim \exp\left(\frac{\lambda_3}{\lambda_1+\lambda_3}\right),
\end{align*}
with $S^{1,X}, S^2, S^3$ mutually independent. We define $S^X:=\min\{S^{1,X},S^2\}$ and $S^Y:=\min\{S^{1,Y},S^3\}$.
\er
\br{bjejej}
If instead of absolute returns we take simple returns, we find that $\lim_{h\downarrow 0} Cov(X_{s+h}-X_s,Y_{t+h}-Y_t)=0$, for any $s\neq t$. Our model is consistent with the fact that empirical cross-correlations of returns are not significant even for very small time lags, as it is for autocorrelations.
\er

From this theorem we obtain an asymptotic evaluation for correlations between absolute log-returns, when the time scale goes to $0$.
Recall that the correlation coefficient between $\xi_s$ and $\eta_t$ is defined as
$$
\rho(\xi_s,\eta_t)=\rho(|X_{s+h}-X_s|,|Y_{t+h}-Y_t|)=\frac{Cov(\xi_s,\eta_t)}{\sqrt{Var(\xi_s)Var(\eta_t)}}.
$$
\begin{corollary}[Decay of cross-asset correlations]
\label{corcorr}
For the process $(X,Y)$ defined above, for any $t> s > 0$, the following expression holds as $h\downarrow 0$: 
\begin{align*}
&\lim_{h\downarrow 0} \rho(\xi_s,\eta_t)=\frac{2}{\pi} \frac{Cov\left((S^X)^{D^X-1/2},(\lambda^Y (t-s)+S^Y)^{D^Y-1/2}\right)}
{\sqrt{Var(|N| S^{D^X-1/2})Var(|N| S^{D^Y-1/2})}}e^{-\lambda^Y (t-s)}
\end{align*}
where with $S$ we denote an exponential variable of parameter $1$ and with $N$ a standard normal variable, and they are mutually independent. $S^X$ and $S^Y$ are defined in Remark \ref{unionpoi}.
\end{corollary} 
\br{samejumps}
Suppose that $X$ and $Y$ are produced by the same time change of two different Brownian motions, i.e $I^X=I^Y=:I$, or:
$$
D^X=D^Y=D,\quad \mathcal{T}^X=\mathcal{T}^Y=\mathcal{T},\quad\bar{\s}^X=\bar{\s}^Y=\bar{\s}.
$$
The expression for the decay of cross-asset correlation in this case is 
$$
\lim_{h\downarrow 0} \rho(\xi_s,\eta_t)=\frac{2}{\pi} \frac{Cov\left(\bar{\s} S^{D-1/2},\bar{\s}(\lambda (t-s)+S)^{D-1/2}\right)e^{-\lambda (t-s)}}
{Var(\bar{\s} |N| S^{D-1/2})},
$$
which is exactly the expression for the decay of autocorrelation  (cf. Theorem \ref{theoautocorr}). This is very close to what we see on empirical  financial data.
\er

\subsection{Empirical results}\label{sectionempiric}
We consider the DJIA Index and FTSE Index, from April 2nd, 1984 to July 6th, 2013. For the data analysis we use the software MatLab \cite{MATLAB:2009}. What follows relies on the ergodicity of the increments $X$ (Proposition \ref{basicpr}-(3)). We start considering the two series separately. We choose some significant quantities related to stylized facts, and use them for the calibration: we consider the multiscaling coefficients $C_1$ and $C_2$, the multiscaling exponent $A(q)$, the volatility autocorrelation function $\rho(t)$. The procedure for the calibration is described precisely in \cite{acdp, tesi_paolo}. 
We find the following estimates for the parameters. 
\begin{align*}
  &\mbox{FTSE:}\quad
  \overline{D} \approx 0.16; \quad 
  \overline{\lambda} \approx 0.0019; \quad
  \overline{\bar{\s}} \approx 0.11. \\  
  &\mbox{DJIA:}\quad
  \overline{D} \approx 0.14; \quad 
  \overline{\lambda} \approx 0.0014; \quad
  \overline{\bar{\s}} \approx 0.127. 
\end{align*}
In Figure \ref{multiscaling} we show the empirical multiscaling exponent versus the prediction of our model with these parameters. Our estimate for the multiscaling exponent looks smoothed out by the empirical curve. 
Since a simulation of daily increments of the model yields a graph analogous to the empirical one, this slight inconsistency is likely due to the fact that the theoretical line shows the limit for $h\downarrow 0$, whereas the empirical data come from a daily sample. 

Figure \ref{autocorrelation} concerns volatility autocorrelation. The decay is between polynomial and exponential, and fits very well empirical data considering that they are quite widespread.

\begin{figure}[!ht]
\centering
\caption{Multiscaling exponent}
\subfigure[FTSE]{
\includegraphics[width=0.46\textwidth]{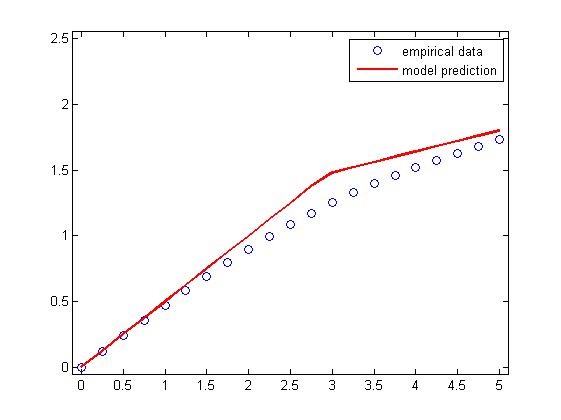}}
\hspace{1mm}
\subfigure[DJIA]{
\includegraphics[width=0.46\textwidth]{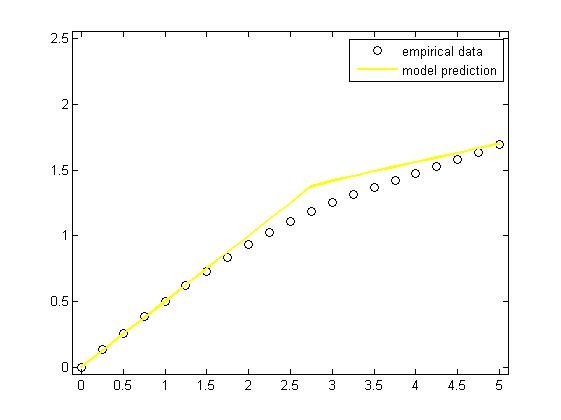}}
\label{multiscaling}
\end{figure}

\begin{figure}[!ht]
\centering
\caption{Volatility autocorrelation}
\subfigure[FTSE log plot]{
\includegraphics[width=0.46\textwidth]{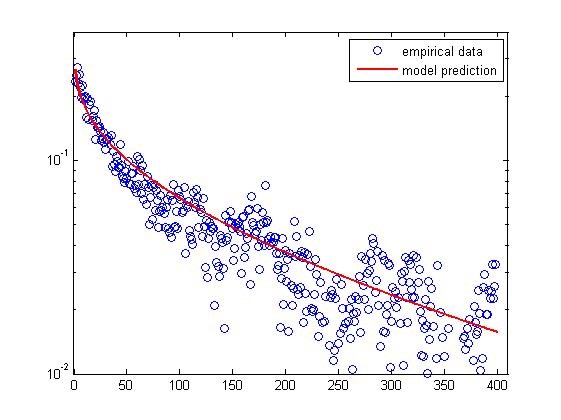}}
\hspace{1mm}
\subfigure[DJIA log plot]{
\includegraphics[width=0.46\textwidth]{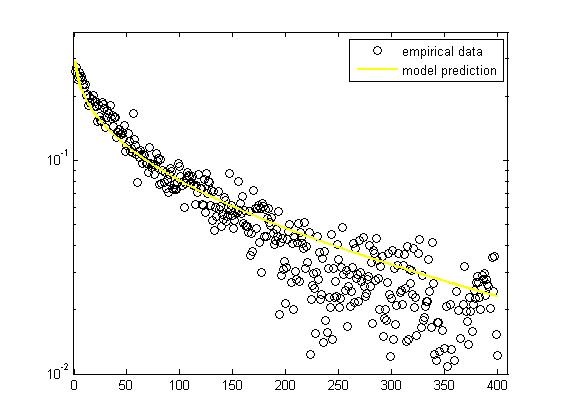}}\\
\subfigure[FTSE loglog plot]{
\includegraphics[width=0.46\textwidth]{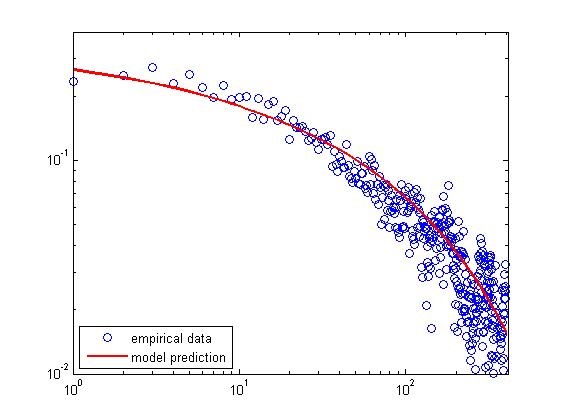}}
\hspace{1mm}
\subfigure[DJIA loglog plot]{
\includegraphics[width=0.46\textwidth]{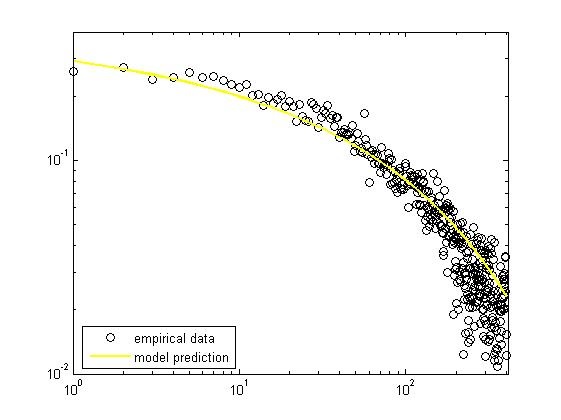}}
\label{autocorrelation}
\end{figure}

We display now the distribution of log returns for our model: $p_t(\cdot)=\PR(X_t\in\cdot)=\PR(X_{n+t}-X_n\in\cdot)$ for $t=1$ day, and the analogous empirical quantity. We do not have an explicit analytic expression for $p_t$, but we can easily obtain it numerically. Figure (\ref{distr}) represents the bulks and the integrated tails of the distributions. We see that the agreement is remarkable, given that this curves are a test \emph{a posteriori}, and no parameter has been estimated using these distributions!

\begin{figure}[!ht]
\centering
\caption{Distribution of log returns}
\subfigure[FTSE bulk]{
\includegraphics[width=0.46\textwidth]{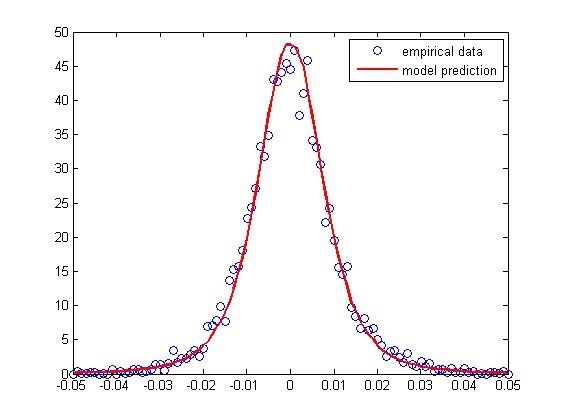}}
\hspace{1mm}
\subfigure[DJIA bulk]{
\includegraphics[width=0.46\textwidth]{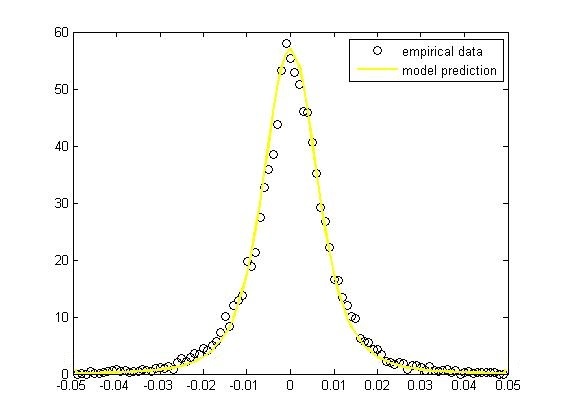}}\\
\subfigure[FTSE integrated tails]{
\includegraphics[width=0.46\textwidth]{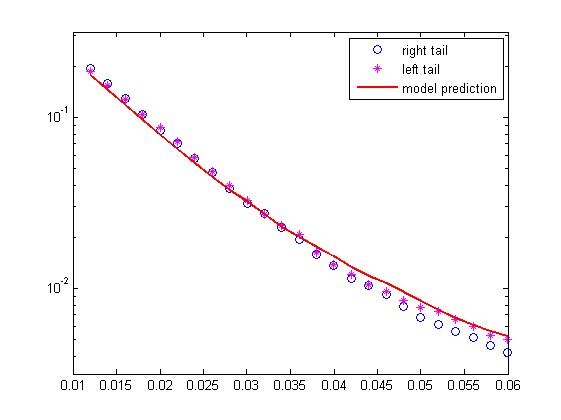}}
\hspace{1mm}
\subfigure[DJIA integrated tails]{
\includegraphics[width=0.46\textwidth]{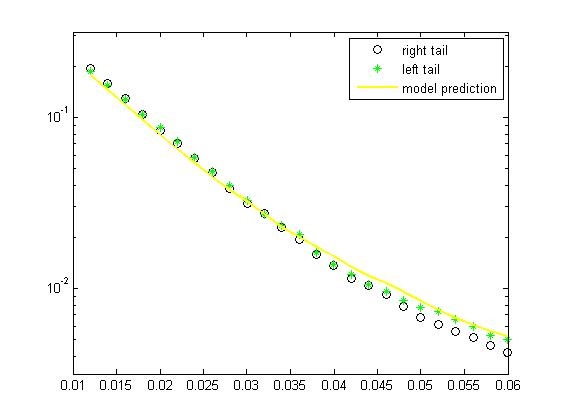}}
\label{distr}
\end{figure}

We have seen in remark \ref{rs} that our estimates for shocks in FTSE and DJIA are strictly related: the occurrence of a shock in an index comes very often together with the occurrence of a shock in the other one.
As a consequence, a first idea to try a rough modeling of cross asset correlations is to suppose $\mathcal{T}^f$, jump process for FTSE, and $\mathcal{T}^d$, jump process for DJIA, to be the same process. From Remark \ref{samejumps} and from the fact that $D$ and $\bar{\s}$ are very similar for FTSE and DJIA, we expect the decay of volatility autocorrelation in the DJIA, the decay of volatility autocorrelation in the FTSE and the decay of cross-asset correlation of absolute returns to display a similar behavior. This is exactly what happens if we plot these quantities (see Figure \ref{emp_corr}), in agreement with the empirical findings of \cite{Podobnik:2007}.
\begin{figure}[!ht]
\centering
\caption{Comparison of empirical correlations}
\subfigure[log plot; one point out of three is plotted]{
\includegraphics[width=0.46\textwidth]{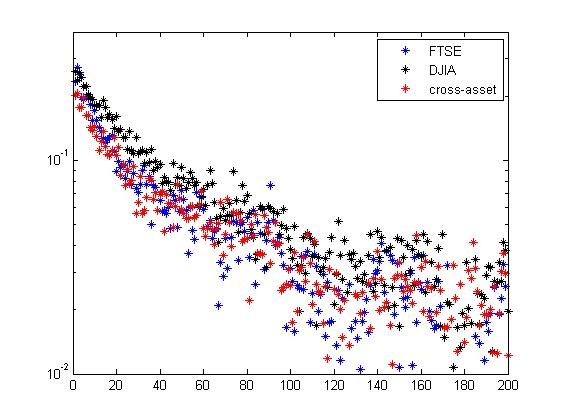}}
\hspace{1mm}
\subfigure[loglog plot; for $t\geq20$, one point out of three is plotted]{
\includegraphics[width=0.46\textwidth]{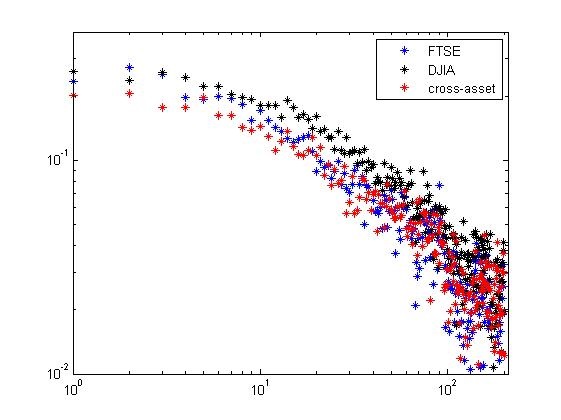}}
\label{emp_corr}
\end{figure}

Under this rough hypothesis, our estimate for cross-asset correlations coincides with our prediction for the decay of volatility autocorrelation in FTSE or DJIA, or with a mean between the two. 
We can do better using the bivariate jump process $I=(I^X,I^Y)$ described at the beginning of section \ref{sectcorr}. We need to estimate the intensities $\lambda_1,\lambda_2,\lambda_3$, subject to the constraints due to the estimated parameters of the one-dimensional models.
Define $\widehat{\gamma}_h(t)$ as the empirical correlation coefficient over h days: 
$$
\widehat{\gamma}_h(t)=corr(|x^f_{\cdot+h}-x^f_\cdot|,|x^d_{\cdot+t+h}-x^d_{\cdot+t}|).
$$
where $x^f$ and $x^d$ are the FTSE and DJIA series of detrended log-returns.
Minimizing a suitable $L^2$ distance between this quantity and the theoretical cross-correlation (Theorem \ref{theoautocorr}
) we obtain
$$
\lambda_1 = 0.0014;\quad \lambda_2 = 0.0005;\quad \lambda_3 = 0.
$$
In Figure \ref{crossassetcorr} we plot the prediction of our model versus the empirical decay of the cross-asset correlations, for $t=1,..,400$ days.

\begin{figure}[!ht]
\centering
\caption{FTSE and DJIA cross-asset correlations}
\subfigure[log plot]{
\includegraphics[width=0.46\textwidth]{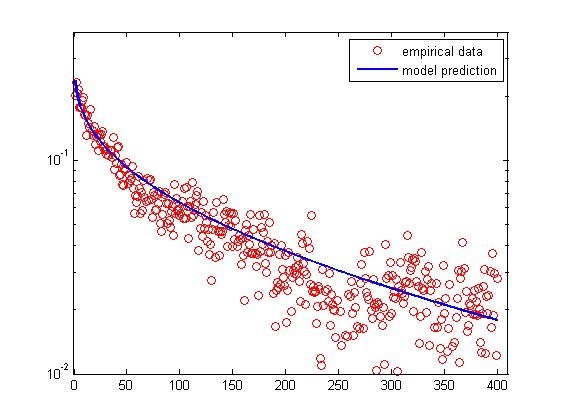}}
\hspace{1mm}
\subfigure[loglog plot]{
\includegraphics[width=0.46\textwidth]{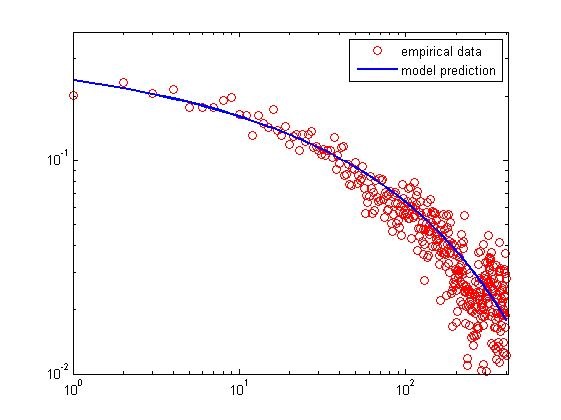}}
\label{crossassetcorr}
\end{figure}

The fact that our estimate is $\lambda_3 = 0$ means that our best fit with real data is obtained when the shocks for the FTSE are given by the shocks of the DJIA plus some additional ones, given by a sparser Poisson process. These estimates, due to the small sample size, are too rough to allow more quantitative considerations, but one may suppose that shocks in the DJIA index always determine a shock in the FTSE index, whereas there might be a shock in FTSE which does not imply a significant increment in the empirical variance of DJIA.

\section{Technical results}
\label{sectproofs}

\subsection{Geometric properties of quadratic variation}
Recall that we have defined the backward difference quotient
\[
Q_{T}(t) = \frac{I_{T} - I_{T- t}}{t}, \quad \mbox{ and }\quad I_t =\langle X \rangle_t.
\]
Considered pathwise, $Q_T(\cdot)$ has some nice geometrical properties which guarantee the existence of an ``isolated'' maximum at the last jump time before $T$. In order to simplify the notation we set 
\be{ma}
m := T - \tau_{i(T)}, \quad \quad 
\alpha := \tau_{i(T)} - \tau_{i(T)-1}, \quad \quad 
K=\big(\frac{1}{2D}\big)^{\frac{1}{1-2D}}.
\ee
The following lemma shows that if $m$ is small enough $Q_T(\cdot)$ attains its maximum at $m$. Moreover, reducing $m$, the value of $Q_T(\cdot)$ at $m$ is arbitrarily larger than the value at the successive mimimum. 
\begin{lemma}
  Let  $m,\,\alpha,K$ as in \eqref{ma}. If $m < K\alpha$, then 
  \begin{enumerate}
    \item $m$ is a local maximum for $Q_T(\cdot)$
    \item 
    $Q_T(\cdot)$ attains its mimimum on $ (m, m+\a )$ at a point $\gamma$.  Moreover,  $Q_T(\cdot)$ is increasing on $(0,m)$, decreasing on $( m, \gamma )$.
    \item The following limit holds 
      \[
	Q_T(m) -  Q_T(\gamma)\xrightarrow{m \to 0+} + \infty  
	   \]
  \end{enumerate}
  \label{lemma:geom}
\end{lemma}

\begin{proof}
    \noindent (1) $Q_T(t)$ is everywhere continuous and it is differentiable but at $\{T - \tau_n\}_{n \in \N}$. To prove that it attains the
      maximum at $m=T - \tau_{i(T)}$ we prove that at $m$ the left
      derivative is greater than $0$ and the right one is less than $0$. 
The derivatives are
\[	
Q'_{T}(t) = \frac{\bar{\s}^{2}2D(T-t-\tau_{i(T)-1})^{2D - 1}t - (I_{T} - I_{T -t}) }{t^2}, \quad\quad t \in ( m, m+\a )
\]
and
\begin{equation}
	\label{eq:derf2}
	Q'_{T}(t) = \frac{\bar{\s}^{2}2D(T-t-\tau_{i(T)})^{2D - 1}t - (I_{T} - I_{T -t}) }{t^2}, \quad\quad t \in ( 0, m )
\end{equation}
      Since $I_s$ is piece-wise concave,
      \[
	I_T - I_{T-t} < I'(T-t)t
      \]
      From \eqref{eq:derf2} we get $  Q'_{T}(t) > 0 $ in $(  0,m )$.  On the other hand 
      \begin{equation}
	\label{eq:limder}
	\lim_{t \to m}\frac{ \bar{\s}^{2}2D(T-t-\tau_{i(T)-1})^{2D - 1}t - (I_{T} - I_{T -t}) }{t^2} = \frac{\bar{\s}^{2}}{m^{2}}\left(2D\alpha^{2D-1}m - m^{2D} \right)=: L_{\bar{\s}}(\alpha,m) 
      \end{equation}                 

      $L_{\bar{\s}}(\alpha,m)$  has the following properties
      \[
\begin{split}
	L_{\bar{\s}}(\alpha,m) = 0 \Leftrightarrow & m = \alpha \left(\frac{1}{2D}\right)^{\frac{1}{1-2D}}\\
	\lim_{m \to 0^{+}} L_{\bar{\s}}(\alpha, m) &= -\infty \\
	\lim_{m \to +\infty} L_{\bar{\s}}(\alpha, m)& = +\infty
\end{split}
\]
      which imply that the right derivative is less than zero if and only if $m <
      K \alpha$, thus $Q_{T}(t)$ attains a local maximum at $m$ if and only if $m < K \alpha$.

\noindent (2)  Note that $Q_T \in \mathcal{C}^{\infty}((m, \a+m ))$ a.s..
      The second order derivative on this interval is
      \[
	Q^{(2)}_T(t) = \frac{\bar{\s}^2 2D(2D-1)(T -t - \tau_{i(T)-1})^{2D - 2}}{t} - \frac{2Q'_T(t)}{t}  
	\label{eq:dersec}
      \]
      Thus $Q'_T(t) = 0$ implies $ Q^{(2)}_T(t) > 0$, then all stationary points are
      minima. Moreover $Q_T$ can have only one minimum which in
      fact exists, since from \eqref{eq:limder} we get
      \[
	\lim_{t \to T - \tau_{i(T)}} Q'_T(t) < 0
      \]
      and
      \[
	\lim_{t \to T - \tau_{i(T)-1}} Q'_T(t) = +\infty 
      \]
      Let $\gamma \in (m, \a+m )$ be the point at  which $Q_T(t)$ attains its minimum. Clearly $Q_T$ is decreasing on $(m,\g)$, and we have already proved that $Q_T$ is increasing on $(0,m)$.

\noindent (3)
      By definition
      \[
	Q_{T}(m) - Q_{T}(\gamma) > Q_{T}(m) - Q_{T}(\alpha + m)
      \]   
      Let $\xi = \tau_{i(T)-1} -  \tau_{i(T)-2}$. We get 
      \[
	Q_{T}(m) - Q_{T}(\alpha + m) = \frac{I_{T} - I_{\tau_{i(T)}}}{T - \tau_{i(T)}} -  \frac{I_{T} - I_{\tau_{i(T)-1}}}{T - \tau_{i(T)-1}} = \bar{\s}^{2}\left(\frac{m^{2D} + \alpha^{2D}}{m} - \frac{m^{2D} + \alpha^{2D} + \xi^{2D}}{m + \alpha}\right)
      \]  
      Passing to the limit
      \[
	Q_{T}(m) - Q_{T}(\alpha + m) = \frac{\bar{\s}^{2}(\alpha^{2D+1} + m^{2D}\alpha - m\xi^{2D})}{m(m + \alpha)} \xrightarrow{m \to 0+} + \infty
      \]   
      then $\lim_{m \to 0+}  Q_{T}(m) - Q_{T}(\gamma) = +\infty$.
\end{proof}

\subsection{Proof of theorem \ref{thm:covar} on  covariance between absolute log-returns}
\label{chaptproofs}
Recall that the increments of $W^X$ and $W^Y$ are independent on disjoint time intervals, and $W^X$ and $\tilde{W}$ are independent Brownian Motions. So for $h<t-s$
\begin{align*}
Cov(\xi_s,\eta_t)&=
\E(|X_{s+h}-X_s| |Y_{t+h}-Y_t|)-\E|X_{s+h}-X_s| \E|Y_{t+h}-Y_t|\\
&= \E\left(|W^X_1| \sqrt{I^X_{s+h}-I^X_s} |\tilde{W}_1| \sqrt{I^Y_{t+h}-I^Y_t}\right)\\
&- \E\left(|W^X_1| \sqrt{I^X_{s+h}-I^X_s}\right) \E \left(|\tilde{W}_1|\sqrt{I^Y_{t+h}-I^Y_t}\right)
\end{align*}
and using independence
\begin{align*}
Cov(\xi_s,\eta_t)&=
(\E|W^X_1|)^2 Cov\left(\sqrt{I^X_{s+h}-I^X_s},\sqrt{I^Y_{t+h}-I^Y_t}\right)\\
&=\frac{2}{\pi} Cov\left(\sqrt{I^X_{s+h}-I^X_s},\sqrt{I^Y_{t+h}-I^Y_t}\right).
\end{align*}
From our choice of $\mathcal{T}^X$ and $\mathcal{T}^Y$ we have the stationarity of the increments of $(I^X,I^Y)$, therefore 
$$
Cov\left(\sqrt{I^X_{s+h}-I^X_s},\sqrt{I^Y_{t+h}-I^Y_t}\right)=
Cov\left(\sqrt{I^X_{h}},\sqrt{I^Y_{t-s+h}-I^Y_{t-s}}\right).
$$
Recall
$$
I_h=\bar{\s}^2\left[(h-\tau_{i(h)})^{2D}+\sum^{i(h)}_{k=1}(\tau_k-\tau_{k-1})^{2D}
-(-\tau_0)^{2D}\right]
$$
Almost surely, for $h$ small enough, $i(h)=i(0)=0$, so the sum in the right hand vanishes and a.s.
\begin{align*}
\lim_{h\downarrow 0}\frac{I_h}{h}&=
\lim_{h\downarrow 0}\bar{\s}^2\frac{(h-\tau_{i(h)})^{2D}-(-\tau_0)^{2D}}{h}\\
&=\bar{\s}^2\lim_{h\downarrow 0}\frac{(h-\tau_0)^{2D}- (-\tau_0)^{2D}}{h}=2D\bar{\s}^2(-\tau_0)^{2D-1},
\end{align*}
and analogously 
\begin{align*}
\lim_{h\downarrow 0}\frac{I_{t+h}-I_t}{h}=2D\bar{\s}^2(t-\tau_{i(t)})^{2D-1}.
\end{align*}
Next Lemma \ref{unif} implies the uniform integrability of the families 
$$
\left\{ \frac{I^X_h}{h}:h\in(0,1] \right\}, \quad
\left\{ \frac{I^Y_{t+h}-I^Y_t}{h}:h\in(0,1] \right\},
$$
therefore we first apply bi-linearity of covariance and then take the limit inside, obtaining
\begin{equation*}
\begin{split}
&\lim_{h\downarrow 0}\frac{Cov\left(\sqrt{I^X_h},\sqrt{I^Y_{t+h}-I^Y_t}\right)}{h}=Cov\left(\lim_{h\downarrow 0}\sqrt{\frac{I^X_h}{h}},\lim_{h\downarrow 0}\sqrt{\frac{I^Y_{t+h}-I^Y_t}{h}}\right)\\
&\quad=2\sqrt{D^X D^Y} \bar{\s}^X\bar{\s}^Y Cov\left((-\tau^X_0)^{D^X-1/2},(t-\tau^Y_{i^Y(t)})^{D^Y-1/2}\right).
\end{split}
\end{equation*}
We can obtain a better representation of this quantity multiplying the right term in the covariance by the characteristic function of $\{i^Y(t)=0\}$ plus the characteristic function of its complement:
\begin{align*}
Cov\left((-\tau^X_0)^{D^X-1/2},(t-\tau^Y_{i^Y(t)})^{D^Y-1/2}\right)
&= Cov\left((-\tau^X_0)^{D^X-1/2},(t-\tau^Y_{i^Y(t)})^{D^Y-1/2}\textbf{1}_{\{i^Y(t)=0\}}\right)\\
&+ Cov\left((-\tau^X_0)^{D^X-1/2},(t-\tau^Y_{i^Y(t)})^{D^Y-1/2}\textbf{1}_{\{i^Y(t)> 0\}}\right)
\end{align*}
The second summand is $0$ because 
$(t-\tau^Y_{i^Y(t)})^{D^Y-1/2}\textbf{1}_{\{i^Y(t)> 0\}}$ is $\mathcal{G}^Y_{>0}$
measurable, where $\mathcal{G}^Y_{>0}=\bar{\s}(\tau^Y_k:\, k>0)$, and $\mathcal{G}^Y_{>0}$ is independent of $\tau_0$. So, using the fact that $\textbf{1}_{\{i^Y(t)=0\}}$ is $\mathcal{G}^Y_{>0}$ measurable, because so is $\textbf{1}_{\{i^Y(t)>0\}}$, we have
\begin {align*}
Cov\left((-\tau^X_0)^{D^X-1/2},(t-\tau^Y_{i^Y(t)})^{D^Y-1/2}\right)
&=Cov\left((-\tau^X_0)^{D^X-1/2},(t-\tau^Y_0)^{D^Y-1/2}\textbf{1}_{\{i^Y(t)=0\}}\right)\\
&=Cov\left((-\tau^X_0)^{D^X-1/2},(t-\tau^Y_0)^{D^Y-1/2}\right)\E\left(\textbf{1}_{\{i^Y(t)=0\}}\right)\\
&=Cov\left((-\tau^X_0)^{D^X-1/2},(t-\tau^Y_0)^{D^Y-1/2}\right)e^{-\lambda^Y t},
\end{align*}
and the theorem is proved. \qed

The following result is used in the proof of Theorem \ref{thm:covar}. Recall that $0<D<1/2$.
\begin{lemma}\label{unif}
The class of random variables 
$$
\left\{\frac{I^X_h}{h}:h\in(0,1]\right\}
$$ 
is bounded in $L^\delta$ for $\delta<\frac{1}{1-2D}$.
\end{lemma}
\begin{proof}
Recall
$$
I_t=\bar{\s}^2\left[(t-\tau_{i(t)})^{2D}+\sum^{i(t)}_{k=1}(\tau_k-\tau_{k-1})^{2D}
-(-\tau_0)^{2D}\right]
$$
and decompose $\E(I_t^\delta)$
$$
\E(I_t^\delta)=\E(I_t^\delta|i(t)=0)\PR(i(t)=0)+\sum_{k=1}^\infty \E(I_t^\delta|i(t)=k)\PR(i(t)=k)
$$
Conditioning on $i(t)=0$ and using convexity,
$$
I_t = \bar{\s}^2\left[(t-\tau_0)^{2D}-(-\tau_0)^{2D}\right] \leq 2D \bar{\s}^2 (-\tau_0)^{2D-1} t
$$
in a right neighborhood of $t=0$. 
So
$$
\E(I_t^\delta|i(t)=0) \leq (2D)^\delta \bar{\s}^{2\delta}\E\left((-\tau_0)^{\delta(2D-1)}\right) t^\delta
\leq C_0 t^\delta
$$
for $\delta<\frac{1}{1-2D}$, since $-\tau_0$ is an random variable with exponential distribution.
Conditioning on $i(t)=k$, $k\geq 1$, and using convexity again,
\begin{align*}
I_t
&\leq \bar{\s}^2\left[ (t-\tau_k)^{2D}+\sum^k_{j=2}(\tau_j-\tau_{j-1})^{2D}+
(t-\tau_0)^{2D}-(-\tau_0)^{2D}\right]\\
&\leq \bar{\s}^2\left[(t-\tau_k)^{2D}+\sum^k_{j=2}(\tau_j-\tau_{j-1})^{2D}+
2D (-\tau_0)^{2D-1} t\right].
\end{align*}
By Jensen inequality and the fact that $2D<1$, 
\begin{align*}
(t-\tau_k)^{2D}+\sum^k_{j=2}(\tau_j-\tau_{j-1})^{2D}\leq k\,\left(\frac{(t-\tau_k)+\sum^k_{j=2}(\tau_j-\tau_{j-1})}{k}\right)^{2D}\leq  k\left(\frac{t}{k}\right)^{2D}
\end{align*}
Then
$$
I_t \leq \bar{\s}^2 \left( 2D(-\tau_0)^{2D-1}t+k\left(\frac{t}{k}\right)^{2D}\right).
$$
Now, supposing $t\leq 1$, we have that for suitable positive constants $C_1$ and $C_2$
\begin{align*}
\E(I_t^\delta|i(t)=k)\PR(i(t)=k)\leq C_1  \frac{\lambda^k}{k!}t^\delta + 
C_2 k^{\delta(1-2D)} \frac{\lambda^k}{k!}t^{1+2D\delta}.
\end{align*}
Recall $\delta<\frac{1}{1-2D}$. Therefore $\delta < 1+2D\delta$, so $t^{1+2D\delta}\leq t^\delta$, and then
$$
\E(I_t^\delta|i(t)=k)\PR(i(t)=k)\leq \left(C_1+C_2 k^{\delta(1-2D)}\right)\frac{\lambda^k}{k!}t^\delta.
$$
Therefore
$$
\E(I_t^\delta)\leq 
\left[ C_0 +\sum_{k=1}^\infty C_3 \frac{\lambda^k}{k!}\right] t^\delta
\leq C_4 t^\delta
$$
where $C_3$ and $C_4$ are positive constants. 
So $\left\{\frac{I^X_t}{t}:t\in(0,1] \right\}$ is bounded in $L^\delta$.
\end{proof}

\section{Acknowledgements}  
We thank Paolo Dai Pra for his constant guidance and advice.

\bibliographystyle{abbrv}
\bibliography{bibliografia}

\begin{thebibliography}{10}

\bibitem{acdp}
A.~Andreoli, F.~Caravenna, P.~D. Pra, and G.~Posta.
\newblock Scaling and multiscaling in financial series: a simple model.
\newblock {\em Advances in Applied Probability}, 44(4):1018--1051, 2012.

\bibitem{BNS}
O.~Barndorff-Nielsen and N.~Shephard.
\newblock Econometric analysis of realized volatility and its use in estimating
  stochastic volatility models.
\newblock {\em Journal of the Royal Statistical Society Series B},
  64(2):253--280, 2002.

\bibitem{Barndorff-Nielsen:2001}
O.~E. Barndorff-Nielsen and N.~Shephard.
\newblock Non-gaussian ornstein-uhlenbeck-based models and some of their uses
  in financial economics.
\newblock {\em Journal of the Royal Statistical Society Series B},
  63(2):167--241, 2001.

\bibitem{Bollerslev:1986}
T.~Bollerslev.
\newblock Generalized autoregressive conditional heteroskedasticity.
\newblock {\em Journal of Econometrics}, 31:307--327, 1986.

\bibitem{tesi_mario}
{Bonino, M.}
\newblock Portfolio allocation and monitoring under volatility shocks.
\newblock Master's thesis, Universit\`a degli Studi di Padova, {2011}.

\bibitem{Brownlees:2012}
C.~T. Brownlees and R.~F. Engle.
\newblock Volatility, correlation and tails for systemic risk measurement.
\newblock {\em Available at SSRN 1611229}, 2012.

\bibitem{tesi_matteo}
{Camelia, M.}
\newblock Variazione quadratica di processi stocastici: un'applicazione alla
  finanza.
\newblock {\em Undergraduate thesis.}, page Universit\`a degli Studi di Padova,
  {2012}.

\bibitem{articolomultiscaling}
P.~Dai~Pra and P.~Pigato.
\newblock Multi-scaling of moments in stochastic volatility models.
\newblock {\em Stochastic Process. Appl.}, 125:3725--3747, 2015.

\bibitem{Embrechts:2002}
P.~Embrechts, A.~McNeil, and D.~Straumann.
\newblock Correlation and dependence in risk management: properties and
  pitfalls.
\newblock {\em Risk management: value at risk and beyond}, pages 176--223,
  2002.

\bibitem{fasen2006extremal}
V.~Fasen, C.~Kl{\"u}ppelberg, and A.~Lindner.
\newblock Extremal behavior of stochastic volatility models.
\newblock In {\em Stochastic finance}, pages 107--155. Springer, 2006.

\bibitem{gray}
S.~F. Gray.
\newblock Modeling the conditional distribution of interest rates as a
  regime-switching process.
\newblock {\em Journal of Financial Economics}, 42(1):27--62, Sept. 1996.

\bibitem{he_maheu}
Z.~He and J.~M. Maheu.
\newblock Real time detection of structural breaks in garch models.
\newblock {\em Computational Statistics \& Data Analysis}, 54(11):2628--2640,
  2010.

\bibitem{icss}
C.~Incl\'{a}n and G.~C. Tiao.
\newblock {Use of cumulative sums of squares for retrospective detection in the
  changes fo variance}.
\newblock {\em J. Amer. Statist. Assoc.}, 89:913--923, 1994.

\bibitem{JacodProtter}
J.~Jacod and P.~E. Protter.
\newblock {\em Discretization of processes}.
\newblock Stochastic modelling and applied probability. Springer, Berlin,
  Heidelberg, 2012.

\bibitem{kluppelberg2006continuous}
C.~Kl{\"u}ppelberg, A.~Lindner, and R.~Maller.
\newblock Continuous time volatility modelling: Cogarch versus
  ornstein--uhlenbeck models.
\newblock In {\em From stochastic calculus to mathematical finance}, pages
  393--419. Springer, 2006.

\bibitem{MATLAB:2009}
MATLAB.
\newblock {\em version 7.8.0 (R2009a)}.
\newblock The MathWorks Inc., Natick, Massachusetts, 2009.

\bibitem{tesi_paolo}
{Pigato, P.}
\newblock A multivariate model for financial indexes subject to volatility
  shocks.
\newblock Master's thesis, Universit\`a degli Studi di Padova, {2011}.

\bibitem{Podobnik:2007}
B.~Podobnik, D.~F. Fu, H.~E. Stanley, and P.~C. Ivanov.
\newblock Power-law autocorrelated stochastic processes with long-range
  cross-correlations.
\newblock {\em The European Physical Journal B - Condensed Matter and Complex
  Systems}, 56(1):47--52, 2007.

\bibitem{Podobnik:2010}
B.~Podobnik, D.~Wang, D.~Horvatic, I.~Grosse, and H.~E. Stanley.
\newblock Time-lag cross-correlations in collective phenomena.
\newblock {\em EPL (Europhysics Letters)}, 90(6):68001, 2010.

\bibitem{ross}
G.~J. Ross.
\newblock Modelling financial volatility in the presence of abrupt changes.
\newblock {\em Physica A: Statistical Mechanics and its Applications},
  392(2):350 -- 360, 2013.

\bibitem{Podobnik:2011}
D.~Wang, B.~Podobnik, D.~Horvati\'c, and H.~E. Stanley.
\newblock Quantifying and modeling long-range cross-correlations in multiple
  time series with applications to world stock indices.
\newblock {\em Phys. Rev. E}, 83(046121), 2011.

\end{thebibliography}

\end{document}